\documentclass[twoside,leqno]{article}
\usepackage[letterpaper]{geometry}

\usepackage[T1]{fontenc}
\usepackage{amsfonts}
\usepackage{graphicx}
\usepackage{epstopdf}
\ifpdf
  \DeclareGraphicsExtensions{.eps,.pdf,.png,.jpg}
\else
  \DeclareGraphicsExtensions{.eps}
\fi

\usepackage{amsopn}

\usepackage{typearea}
\usepackage{setspace}

\usepackage{fullpage}

\usepackage{comment} 
\usepackage{bbm}
\usepackage{framed} 
\usepackage{url} 
\usepackage{complexity}
\usepackage{booktabs}
\usepackage{amsmath,amssymb}
\usepackage{float}

\usepackage{amsthm}
\usepackage{thmtools} 
\usepackage{thm-restate}
\usepackage{nicefrac}
\usepackage{calc}
\usepackage{enumerate}
\usepackage{enumitem}
\usepackage[usenames,dvipsnames]{xcolor}
\usepackage{dsfont}

\usepackage[ruled,vlined,linesnumbered]{algorithm2e}
\usepackage{forloop}

\usepackage{graphicx}
\usepackage[font=footnotesize,labelfont=bf]{subcaption}
\usepackage[font=footnotesize,labelfont=bf]{caption}
\usepackage{appendix}
\usepackage[noend]{algpseudocode}
\usepackage[colorinlistoftodos]{todonotes}
\usepackage{xr}
\usepackage{array}
\usepackage{xspace}
\definecolor{ForestGreen}{rgb}{0.1333,0.5451,0.1333}
\definecolor{DarkRed}{rgb}{0.8,0,0}
\definecolor{Red}{rgb}{1,0,0}
\usepackage[linktocpage=true,
pagebackref=true,colorlinks,
linkcolor=DarkRed,citecolor=ForestGreen,
bookmarks,bookmarksopen,bookmarksnumbered]{hyperref}

\usepackage[capitalise]{cleveref}
\crefrangelabelformat{enumi}{#3#1#4--#5#2#6}
\usepackage{parskip}
\usepackage{chngcntr}
\usepackage{mathtools,stackengine}
\usepackage{multirow}

\stackMath
\newcommand{\stackGeq}[1]{%
	\setbox0=\hbox{${}\mathrel{\stackon[-1pt]{\geq}{\scriptstyle\text{#1\strut}}}{}$}
	\xdef\tmpwd{\dimexpr\the\wd0\relax}
	\kern.5\tmpwd\mathclap{\box0}&\kern.5\tmpwd
}

\usepackage{nameref}

\usepackage{tikz}
\usetikzlibrary{patterns}

\allowdisplaybreaks

\newcommand\eps{\epsilon}

\DeclarePairedDelimiterX{\expectarg}[1]{[}{]}{%
	\ifnum\currentgrouptype=16 \else\begingroup\fi
	\activatebar#1
	\ifnum\currentgrouptype=16 \else\endgroup\fi
}

\DeclarePairedDelimiterX{\nicesetarg}[1]{\{}{\}}{%
	\ifnum\currentgrouptype=16 \else\begingroup\fi
	\activatebar#1
	\ifnum\currentgrouptype=16 \else\endgroup\fi
}

\newcommand{\innermid}{\nonscript\;\delimsize\vert\nonscript\;}
\newcommand{\activatebar}{%
	\begingroup\lccode`\~=`\|
	\lowercase{\endgroup\let~}\innermid 
	\mathcode`|=\string"8000
}

\newcommand\opt{\textsc{Opt}\xspace}

\newcommand\alg{\textsc{Alg}\xspace}

\counterwithin{equation}{section}

\usepackage{eqparbox}

\theoremstyle{plain}

\newtheorem{theorem}{Theorem}[section]

\newtheorem{proposition}[theorem]{Proposition}

\newtheorem{lemma}[theorem]{Lemma}
\newtheorem{fact}[theorem]{Fact}
\newtheorem{observation}[theorem]{Observation}

\newlength{\continueindent}
\usepackage{etoolbox}
\makeatletter
\newcommand*{\ALG@customparshape}{\parshape 2 \leftmargin \linewidth \dimexpr\ALG@tlm+\continueindent\relax \dimexpr\linewidth+\leftmargin-\ALG@tlm-\continueindent\relax}
\apptocmd{\ALG@beginblock}{\ALG@customparshape}{}{\errmessage{failed to patch}}
\makeatother

\makeatletter
\def\thm@space@setup{%
	\thm@preskip=\parskip \thm@postskip=0pt
}
\makeatother

\usepackage{etoolbox}
\usepackage{tikz}
\usetikzlibrary{tikzmark}
\usetikzlibrary{calc}

\errorcontextlines\maxdimen

\newcommand{\ALGtikzmarkcolor}{black}
\newcommand{\ALGtikzmarkextraindent}{4pt}
\newcommand{\ALGtikzmarkverticaloffsetstart}{-.5ex}
\newcommand{\ALGtikzmarkverticaloffsetend}{-.5ex}
\makeatletter
\newcounter{ALG@tikzmark@tempcnta}

\newcommand\ALG@tikzmark@start{%
	\global\let\ALG@tikzmark@last\ALG@tikzmark@starttext%
	\expandafter\edef\csname ALG@tikzmark@\theALG@nested\endcsname{\theALG@tikzmark@tempcnta}%
	\tikzmark{ALG@tikzmark@start@\csname ALG@tikzmark@\theALG@nested\endcsname}%
	\addtocounter{ALG@tikzmark@tempcnta}{1}%
}

\def\ALG@tikzmark@starttext{start}
\newcommand\ALG@tikzmark@end{%
	\ifx\ALG@tikzmark@last\ALG@tikzmark@starttext
	\else
	\tikzmark{ALG@tikzmark@end@\csname ALG@tikzmark@\theALG@nested\endcsname}%
	\tikz[overlay,remember picture] \draw[\ALGtikzmarkcolor] let \p{S}=($(pic cs:ALG@tikzmark@start@\csname ALG@tikzmark@\theALG@nested\endcsname)+(\ALGtikzmarkextraindent,\ALGtikzmarkverticaloffsetstart)$), \p{E}=($(pic cs:ALG@tikzmark@end@\csname ALG@tikzmark@\theALG@nested\endcsname)+(\ALGtikzmarkextraindent,\ALGtikzmarkverticaloffsetend)$) in (\x{S},\y{S})--(\x{S},\y{E});%
	\fi
	\gdef\ALG@tikzmark@last{end}%
}

\apptocmd{\ALG@beginblock}{\ALG@tikzmark@start}{}{\errmessage{failed to patch}}
\pretocmd{\ALG@endblock}{\ALG@tikzmark@end}{}{\errmessage{failed to patch}}
\makeatother

\algblock[with]{With}{EndWith}
\algblockdefx[With]{With}{EndWith}%
[1]{\textbf{with} #1 \textbf{do}}%
{}

\makeatletter
\ifthenelse{\equal{\ALG@noend}{t}}%
{\algtext*{EndWith}}
{}%
\makeatother

\newcommand\bE{\mathbb{E}}
\newcommand\cA{\mathcal{A}}

\usepackage[nobreak=true]{mdframed}

\newcounter{customalg}

\crefname{customalg}{algorithm}{algorithms}
\Crefname{customalg}{Algorithm}{Algorithms}

\newenvironment{customalg}[1][]{%
  \par\addvspace{0.1in}
  \refstepcounter{customalg}
  \begin{mdframed}[
    hidealllines=true,
    backgroundcolor=gray!20,
    innerleftmargin=0.35cm,
    innerrightmargin=0.35cm,
    innertopmargin=0.325cm,
    innerbottommargin=0.325cm,
    roundcorner=10pt
  ]
  \noindent\textbf{Algorithm~\thecustomalg}
    \if\relax\detokenize{#1}\relax
    \textbf{:}%
  \else
    \textbf{: #1}%
  \fi
  \par\vspace{1pt}
  \setstretch{0.9}
}{%
  \end{mdframed}%
  \par\addvspace{\topsep}
}

\newcommand\epssig{\scalebox{1.1}{$\epsilon$}_{\scalebox{0.5}{$\Sigma$}}}
\newcommand\epspi{\scalebox{1.1}{$\epsilon$}_{\scalebox{0.5}{$\Pi$}}}

\title{Trading Prophets with Initial Capital}
\author{
 {Yossi Azar\thanks{Department of Computer Science, Tel-Aviv University,
     Tel-Aviv, Israel. Emails: {\tt azar@tauex.tau.ac.il, orvardi@mail.tau.ac.il}.}
    }
 \and
 {Niv Buchbinder\thanks{Department of Statistics and Operations Research, School of Mathematical Sciences, Tel Aviv University, Tel Aviv, Israel. Email: \texttt{niv.buchbinder@gmail.com}. The work of Niv Buchbinder is supported in part by the Israel Science Foundation (ISF) grant no. 3001/24, and the United States - Israel Binational Science Foundation (BSF) grant no. 2022418.}
 }
 \and
 {Roie Levin\thanks{Department of Computer Science, Rutgers University, Piscataway, NJ 08854. Email: {\tt roie.levin@rutgers.edu.}}}
 \and
 {Or Vardi$^{*}$
   } 
}    
\date{}

\begin{document}

\maketitle

\begin{quote}
\textit{``The easiest way to become a millionaire is to be born one.''} -- Folklore
\end{quote}

\begin{abstract} 
    \smallskip
    Correa et al. [EC' 2023] introduced the following
    trading prophets problem. A trader observes a sequence of stochastic prices for a stock, each drawn from a known distribution, and at each time must decide whether to buy or sell.
    Unfortunately, they observed that in this setting it is impossible to compete with a prophet who knows all future stock prices.

    In this paper, we explore the trading prophets problem when we are given initial capital with which to start trading. We show that initial capital is enough to bypass the impossibility result and obtain a competitive ratio of $3$ with respect to a prophet who knows all future  prices (and who also  starts with capital), and we show that this competitive ratio is best possible. We further study a more realistic model in which the trader must pay multiplicative and/or additive transaction costs for trading which model dynamics such as bid-ask spreads and broker fees. 
\end{abstract}

\section{Introduction}

Is there an algorithm for making money in the stock market?

Building on the well-established foundations of optimal stopping theory and prophet inequalities~\cite{Krengel1977SemiamartsAF, prophI-median-2},  Correa et al.\ \cite{CDHOS23} formalized a version of this question which they termed the {\em trading prophets} problem. In their model, a trader observes a daily price for the stock or commodity, where the price $X_i \sim D_i$ is drawn from a known distribution $D_i$. Each day, the trader must decide whether to buy or sell, while holding at most one unit of the asset.\footnote{More generally, the trader may be limited to a fixed number of units, which can be normalized to one for simplicity.}
This constraint may arise from storage limitations, regulatory restrictions, or risk management policies that cap the trader's position. Subject to these constraints, the trader's objective is to maximize profit, defined as the total revenue from sales minus the total cost from purchases, and the benchmark to measure against is the performance of a ``prophet'' who 
has a perfect foresight of all future prices.
    
Correa et al.\ observe a simple example in which the competitive ratio of any algorithm is unbounded. The instance spans two days: on day one the price is $1$, and on day two the price is $1/\eps$ with probability $\eps$ or $0$ otherwise.  The algorithm expects that buying on day one leads to profit $\eps \cdot \nicefrac{1}{\eps} - 1 = 0$, so it might as well never buy and sit the game out entirely.  The prophet on the other hand, with perfect foresight, chooses to buy on day one only when profit is guaranteed on day two, and this grants him strictly positive expected returns of $\eps \cdot (\nicefrac{1}{\eps} - 1) > 0$.
As a result, Correa et al.\ focus on a relaxed (and arguably less realistic) setting in which the daily price distributions are either i.i.d.\ or presented in random order. 

At first glance, the example above appears to close the door on the general 
version of the trading prophets problem. However, a closer look reveals that the impossibility of the instance relies on the buy-in cost for the algorithm being prohibitively expensive; in other words, the algorithm must start with no initial capital. One might argue that this is pathological, and that meaningful instances of the trading prophets problem typically span many time steps and enable the trader to amortize buy-in costs over time.

This observation raises the central question we explore in this work:
    \begin{quote}
        \emph{Do there exist competitive algorithms for trading prophets when both the algorithm and the prophet start with initial capital?}
    \end{quote}

\subsection{Our Results}

Our first result is that, perhaps surprisingly, buy-in cost is the \emph{only} barrier to achieving good competitive algorithms for the trading prophets problem. Specifically, if we allow the algorithm and the prophet to begin the game with one share of stock (or more generally, if we allow the algorithm to incur an additive loss on top of the multiplicative competitive ratio), the prophet's advantage is dramatically reduced, and the following simple algorithm is $3$-competitive.

\begin{customalg}[Buy Low Sell High]
\label{alg:buylsellh}
For all time steps $i=0,1,\ldots, T$, observe the price $X_i \sim D_i$.  
 Then,
\begin{itemize}
    \item If the algorithm holds the stock, and $X_i> \bE[X_{i+1}]$, sell the stock.
    \item If the algorithm does not hold the stock, and $X_i\leq \bE[X_{i+1}]$, buy the stock.
\end{itemize} \vspace{-0.05in}
\end{customalg}

\Cref{alg:buylsellh} has the added advantage that it is myopic, in that it does not require full foreknowledge of all future distributions (but only the expected stock price tomorrow). In fact it requires only a single bit of information about the future: whether the expected price of the stock tomorrow higher or lower than the today's price. This algorithm generalizes that of \cite{RCV25} that was designed 
for the i.i.d/random order model (though theirs also works in a more general setting involving multiple stocks).
We extend the observation in \cite{RCV25} and show that this simple strategy is in fact optimal among all online algorithms for the adversarial order trading prophets problem. Our main result is:
\begin{theorem}\label{thm:main1}
    \Cref{alg:buylsellh} is $3$-competitive for the adversarial order trading prophets problem. Moreover, if $\bE[X_i]=\mu$ for all $i=1, \ldots, T$, then the algorithm is $2$-competitive. These competitive ratios are best possible. 
\end{theorem}
It is worth noting that the proof of this theorem 
relies solely on the initial stock granted to both the algorithm and the optimal solution, with no need for any further additive term in the competitive ratio.  \vspace{-0.2in}

\urldef{\BAS}\url{https://en.wikipedia.org/wiki/Bid-ask_spread}
\paragraph{Adding transaction costs or bid-ask spread.}
    In reality, no trader can buy and sell the stock at the same price, and this can happen due to several reasons. Firstly, in trading platforms there is usually a {\em bid-ask spread}, i.e.\ an inherent difference between the cost of buying and selling the same commodity that arises naturally due to market dynamics.\footnote{See, e.g., {\BAS}.} 
    Secondly, the trader often pays a {\em transaction cost} to the broker either as an added percentage of the price or as a fixed flat cost. To model these aspects, we study the {\em trading prophets problem with transaction costs} in which there are two additional parameters $\epspi, \epssig \in [0,1)$. At each time step $i$, the base price is still sampled as $X_i \sim D_i$, but the algorithm can now only buy the stock at a price of  $(1+\epspi)X_i+\epssig$ or sell the stock at a price of $(1-\epspi)X_i-\epssig$. We refer to $\epspi$ and $\epssig$ as the multiplicative and additive transaction costs respectively. 

This seemingly small change to the model makes the trading prophets problem significantly harder. In Appendix \ref{sec:bad-example} we show that even in the special case when $\epspi=0$ and there is only additive transaction cost, no myopic algorithm that at time $i$ uses only information about $D_{i+1}$ can be competitive.\footnote{We actually show the stronger fact that even a myopic algorithm with knowledge of the \emph{realization} of $X_{i+1}$ cannot be competitive.} In fact, even in the i.i.d.\ case, the natural generalization of \Cref{alg:buylsellh} that buys the stock whenever $X_i \leq \bE[X_{i+1}]-\eps$ and sells the stock whenever $X_i\geq \bE[X_{i+1}]+\eps$ fails, and new ideas are needed.
Nevertheless, we show that the following algorithm is $2$-competitive for the i.i.d.\ setting.

 \begin{customalg}[Buy Below Sell Above]
\label{alg:trnsctcost}
Define $z_H\geq z_L$ to be thresholds
such that for all $i \in [T]$,\footnote{Recall that $X_i$ are i.i.d. We show in \Cref{sec:pre} that such values $z_H$ and $z_L$ exist without loss of generality. \vspace{-0.14in}}
 \[Pr[X_i\geq z_H]=Pr[X_i\leq z_L],\]
 and
 \[z_H(1-\epspi)-\epssig = z_L(1+\epspi)+\epssig .\]

For all time steps $i=0,1,\ldots, T$, observe the price $X_i \sim D$.  
 Then,
\begin{itemize}
    \item If the algorithm holds the stock, and $X_i\geq z_H$, sell the stock.
    \item If the algorithm does not hold the stock, and $X_i\leq z_L$, buy the stock.
\end{itemize} 
\end{customalg}

Our algorithm is again myopic and only needs to know the two threshold values $z_L\leq z_H$ that depend on the distribution. 
We note that the median value $M$, i.e. the value for which $Pr[X\geq M]=Pr[X\leq M]=\frac{1}{2}$, always satisfies  $z_L \le M \le z_H$  (but $M$ need not be precisely in the middle of the range). When $\epspi=\epssig=0$ then $z_H=z_L=M$ and our algorithm reduces to the algorithm analyzed in \cite{CDHOS23} for the i.i.d.\ and random arrival models. Our second result is:
\begin{theorem}\label{thm:main2}
\Cref{alg:trnsctcost} is $2$-competitive for the trading prophets problem with transaction costs in the i.i.d.\ model. 
\end{theorem}

We leave as an intriguing open question whether there exists a constant competitive algorithm for the trading prophets problem with transaction costs in the adversarial order model.

\subsection{Related Work}

\paragraph{Online Trading.}

Correa et al.\ \cite{CDHOS23} initiated the study of the trading prophets problem. They obtained a tight $2$-competitive algorithm for the i.i.d.\ model.  They also obtained a $16$-competitive algorithm for the non-i.i.d.\ random arrival model. Following their work Rajput, Chiplunkar and Vaish~\cite{RCV25} studied a generalization of the basic setting in which there are multiple stocks whose prices at a given time may be correlated. The algorithm can hold only a subset of the stocks, or more generally there are matroid constraints on the feasible subsets of stocks that it can hold. They designed algorithms for this more general setting in the i.i.d and random arrival models.

Very recently, Azar et al.\ \cite{azar2025competitivebundletrading}
studied a significantly more general model in which the trader can buy and sell $n$ types of goods in bundles. Unlike \cite{CDHOS23,RCV25}, these authors no longer assume prices are stochastically generated, but instead changing in arbitrary adversarial ways; they also allow buyers and sellers to arrive at separate time steps, which is the same as allowing for unbounded bid-ask spreads. Because of the additional difficulty, Azar et al.\ only managed to obtain logarithmic competitive ratios and their results require resource augmentation, i.e. they compare their algorithm to an optimal solution that must trade at slightly worse prices. It is an interesting question to explore intermediate models that bridge the stochastic and fully adversarial settings.

\paragraph{Prophet Inequalities.} 

The term ``prophet inequality'' was originally used to describe the sell-only analog of our setting. The first such result is due to Krengel and Sucheston \cite{Krengel1977SemiamartsAF} (see also \cite{Hill1982ComparisonsOS} for an early survey). Other applications to computer science were given by Hajiaghayi et al.\ \cite{DBLP:conf/aaai/HajiaghayiKS07} and Chawla et al. \cite{DBLP:conf/stoc/ChawlaHMS10}.

The literature on prophet inequalities is too vast to cover fairly here. We briefly mention some recent trends in this area: blends of prophet and secretary problems ~\cite{EsfandiariHLM17,EhsaniHKS24,CorreaSZ21}, single (or few) sample variants where the algorithm only gets sample access to the distributions $D_i$ \cite{azar14prophet,azar-first-sspi,CaramanisDFFLLP22,DBLP:journals/corr/abs-2304-02063}, and versions that explore correlation structures beyond product distributions \cite{MR0913569,DBLP:journals/tcs/ChekuriL24,DBLP:conf/sigecom/LivanosP024,DBLP:conf/wine/CaragiannisGLW21,DBLP:conf/ipco/GuptaHKL24,DBLP:conf/stoc/DughmiKP24,DBLP:journals/teco/ImmorlicaSW23,DBLP:conf/sigecom/MaurasMR24}.

\section{Preliminaries}\label{sec:pre}

We begin by describing the trading prophets problem formally. A trader observes a sequence of $T$ prices $X_1, X_2, \ldots, X_T$ for a single stock, where $X_i\sim D_i$ is drawn from a non-negative distribution $D_i$ that is known to the trader. 
After the price $X_i$ is revealed, the trader must decide whether to buy or sell at that price, and the objective is to maximize profit. We assume the trader may hold at most one unit of stock at any point in time (see \cite{CDHOS23} for why this is without loss of generality). Define $\mu_i= \bE[X_i]$ be the expected price at time $i$. For convenience, we assume there are extra distributions $X_0$ and $X_{T+1}$ that are both 0 with probability 1: these allow the algorithm and the prophet to purchase a stock initially for a price of $0$, and guarantee our algorithms are well defined at time $T$.

Using the notation $[i,j] := \{i, \ldots, j\}$ for integers $i$ and $j$, define $B\subseteq [0,T+1]$ and $S\subseteq[1,T+1]$ to be the subsets of times in which the algorithm buys and sells the stock respectively. Then the profit of the algorithm is $\alg := \sum_{i\in S}X_i- \sum_{i\in B}X_i$. When we write $\bE[\alg]$, we take the expectation over the random variables $X_1, \ldots, X_i$ (and any randomness of the algorithm). 

Let \opt denote the maximum profit achievable offline with full information about the realizations of $X_1, \ldots, X_T$ and, currently, with no transaction costs. As observed in \cite{CDHOS23}, this optimal solution always buys at local minima (time steps for which $X_i\leq \max(X_{i-1},X_{i+1})$) and sells at local maxima (time step for which $X_i\geq \max(X_{i-1},X_{i+1})$), and since successive difference terms telescope, we can write $\opt=\sum_{i=1}^{T}(X_i-X_{i-1})_+$, where $(X_i-X_{i-1})_+ \triangleq \max\{0,(X_i-X_{i-1})\}$. 

\begin{observation}\label{obs-opt}
For the trading prophets problem, 
    $\bE[\opt] =\sum_{i=1}^{T}\bE[(X_i-X_{i-1})_+]$.
\end{observation}

In the more general version of the trading prophets problem with transaction costs, we are given two additional parameters $\epspi, \epssig \in [0,1)$. At each time step the algorithm, as well as the offline prophet, can buy the stock at a price of $(1+\epspi)X_i+\epssig$ and sell the stock at a price of $(1-\epspi)X_i-\epssig$. We observe that in this case it is no longer true that $\bE[\opt] =\sum_{i=1}^{T}\bE[(X_i-X_{i-1})_+]$, because the difference terms no longer telescope.

We say that an algorithm is $\alpha$-competitive if  
\begin{equation}
  \bE[\opt] \leq \alpha \cdot \bE[\alg] +c,  
\end{equation}
where $c$ is a constant independent of the length of the sequence $T$. We remark that allowing such an additive constant is common in online computation and competitive analysis.

To implement our algorithm, we require the following fact.

\begin{fact}
    For any real valued, nonnegative random variable $X$ with continuous cdf, and any $\epspi,\epssig \in [0,1)$, there exist thresholds $z_L$ and $z_H$ such that simultaneously
    \begin{enumerate}[label = (\roman*)]
        \item $\displaystyle Pr[X\geq z_H]=Pr[X\leq z_L]$, and
        \item $\displaystyle z_H(1-\epspi)-\epssig = z_L(1+\epspi)+\epssig$.
    \end{enumerate}
\end{fact}

\begin{proof}
    Given a nonnegative real number $h$, define $\ell(h) = \frac{h \cdot (1-\epspi) - 2\cdot \epssig}{1+\epspi}$, such that $z_H = h$ and $Z_L = \ell(h)$ satisfy (ii). Then define the function
    \[f(h) = \Pr(X \geq h) - \Pr(X \leq \ell(h)).\]
    When $h = 0$, we have $f(h) = 1$, and when $h \rightarrow \infty$, we have $f(h) = -1$, so by the intermediate value theorem, there exists a value $h^*$ such that $z_H = h^*$ and $z_L = \ell(h^*)$ satisfy both (i) and (ii).
\end{proof}
As observed in \cite{CDHOS23}, this assumption is without loss of generality, as we can add to each price an auxiliary independent uniform perturbation $\delta_i \in [-\delta,\delta]$ and taking $\delta$ arbitrary small. We remark that after adding this perturbation, our algorithm may become randomized.

\section{Trading Prophets in the Adversarial Model}\label{sec:adversarial}

In this section we analyze \Cref{alg:buylsellh} for the trading prophets problem, proving \Cref{thm:main1}.

The upper bound on the competitive ratio follows directly by the following two lemmas.
\begin{lemma}\label{lem:alg-basic}
    \Cref{alg:buylsellh} satisfies $\bE[\alg]  = \sum_{i=1}^{T}\bE[(\mu_{i}-X_{i-1})_+]$.
\end{lemma}

\begin{lemma}\label{lem:opt-basic}
It holds that 
    $\bE[\opt] \leq 3 \cdot \sum_{i=1}^{T}\bE[(\mu_{i}-X_{i-1})_+]$. Moreover, if $\bE[X_i]=\mu$ for all $i\in[1,T]$, then $\bE[\opt] \leq 2 \cdot \sum_{i=1}^{T}\bE[(\mu-X_{i-1})_+]$
\end{lemma}

\begin{proof}[Proof of \Cref{lem:alg-basic}]
For any $i\in[0,T]$ define $p_i\triangleq  \Pr(X_i\leq\mu_{i+1})$, $q_i\triangleq \Pr(X_i>\mu_{i+1})=1-p_i$. Let $\cA_i$ be the event that the algorithm holds the stock at time step $i$ (before the price $x_i$ is revealed). We observe that the probability that the algorithm holds the stock at time step $i\in [1,T]$ is exactly $p_{i-1}$. Indeed, if $X_{i-1}\leq \mu_{i}$ (which happens with probability $p_{i-1}$), then the algorithm does not sell the stock (if the algorithm has the stock in inventory), and buys the stock (if the algorithm does not have the stock in inventory). Similarly, if $X_{i-1}> \mu_{i}$ the algorithm sells the stock (if the algorithm has the stock in inventory) and does not buy the stock (if the algorithm does not have the stock in inventory). By the law of total expectation, 
\begin{align}
    \bE[\alg] &=\sum_{i=1}^{T}\bigg( \Pr[\cA_{i-1}] \cdot q_i \cdot \bE[X_i \mid X_i> \mu_{i+1}] - \Pr[\lnot \cA_{i-1}] \cdot p_i \cdot \bE[X_i \mid X_i \leq \mu_{i+1}]\bigg) \nonumber\\
    &=\sum_{i=1}^{T}\bigg(p_{i-1} \cdot q_i \cdot \bE[X_i \mid X_i> \mu_{i+1}] - (1-p_{i-1}) \cdot p_i \cdot \bE[X_i \mid X_i \leq \mu_{i+1}]\bigg) \nonumber\\
    & = \sum_{i=1}^{T}\bigg(p_{i-1} \cdot \bigg(q_i \cdot \bE[X_i \mid X_i> \mu_{i+1}] +  p_i \cdot \bE[X_i \mid X_i \leq \mu_{i+1}]\bigg)-  p_i \cdot \bE[X_i \mid X_i \leq \mu_{i+1}]\bigg) \nonumber\\
    & = \sum_{i=1}^{T}\bigg(p_{i-1}\cdot  \bE[X_i]- p_i \cdot E[X_i \mid X_i\leq \mu_{i+1}]\bigg) \nonumber\\
    & = \sum_{i=1}^{T}p_{i-1}\bigg(\mu_{i}- \bE[X_{i-1} \mid X_{i-1}\leq \mu_{i}]\bigg) \label{eq11} \\
    &= \sum_{i=1}^{T}\bE[(\mu_{i}-X_{i-1})_+]. \nonumber 
\end{align}
\Cref{eq11} holds since
$\bE[X_0 \mid X_0\leq \mu_{1}]=\bE[X_T \mid X_T\leq \mu_{T+1}]=0$ (as $X_0=X_{T+1}=0$).
\end{proof}
We finish with the bound on $\opt$.
\begin{proof}[Proof of \cref{lem:opt-basic}] 
We first observe that for any $i\in[1,T]$ we have,
\begin{align}
\lefteqn{
\bE[(X_{i}-X_{i-1})_+]
 = \bE \big[ \big((X_{i}-\mu_{i+1}) + (\mu_{i+1}-\mu_{i}) + (\mu_{i}-X_{i-1})\big)_+ \big]} \nonumber \\
&\leq  \bE\big[(X_{i}-\mu_{i+1})_+\big]  + (\mu_{i+1}-\mu_{i})_+ + \bE\big[(\mu_{i}-X_{i-1})_+\big] \label{ineq21}\\
&=  \big(\mu_{i}-\mu_{i+1}+\bE\big[(\mu_{i+1}-X_{i})_+\big]\big) + (\mu_{i+1}-\mu_{i})_+ + \bE\big[(\mu_{i}-X_{i-1})_+\big] \label{ineq221}\\
&\leq  \big(\mu_{i}-\mu_{i+1}+\bE\big[(\mu_{i+1}-X_{i})_+\big]\big) + \bE[(\mu_{i+1}-X_{i})_+] + \bE\big[(\mu_{i}-X_{i-1})_+\big] \label{ineq222}\\
& = \mu_{i}-\mu_{i+1}+ 2 \cdot \bE\big[(\mu_{i+1}-X_{i})_+\big] + \bE\big[(\mu_{i}-X_{i-1})_+\big] .\nonumber
\end{align}
Inequality \eqref{ineq21} follows as $(a+b+c)_+\leq a_+ +b_+ +c_+$ (and by linearity of expectation).
Equality \eqref{ineq221} follows as $\bE[(X_{i}-\mu_{i+1})_+]=\mu_{i}-\mu_{i+1} + \bE[(\mu_{i+1}-X_{i})_+]$ because 
\begin{align*}
 \bE[(X_{i}-\mu_{i+1})_+] +\mu_{i+1}-\mu_i = \bE[\max\{X_i-\mu_{i+1}, 0\} + \mu_{i+1}- X_i] =  \bE\big[(\mu_{i+1}-X_{i})_+\big] .
\end{align*}
Inequality \eqref{ineq222} follows as $(\mu_{i+1}-\mu_{i})_+ \leq \bE[(\mu_{i+1}-X_{i})_+]$ since 
\begin{align*}
(\mu_{i+1}-\mu_{i})_+ = \max\{0, \bE[\mu_{i+1}-X_{i}]\} \leq \bE[\max\{0, \mu_{i+1}-X_{i}\}] = \bE[(\mu_{i+1}-X_{i})_+] .
\end{align*} 
By \Cref{obs-opt} summing over all $T\in[1,T]$ and using that $X_0= X_{T+1}=0$, we get
\begin{align}
\bE[\opt]&= \sum_{i=1}^{T}\bE[(X_{i}-X_{i-1})_+] \nonumber\\
& \leq \sum_{i=1}^{T} \bigg(\mu_{i}-\mu_{i+1}+ 2 \cdot \bE\big[(\mu_{i+1}-X_{i})_+\big] + \bE\big[(\mu_{i}-X_{i-1})_+\big]\bigg) \nonumber \\
& = \sum_{i=1}^{T} (\mu_{i}-\mu_{i+1}) + 2 \sum_{i=1}^{T} \bE\big[(\mu_{i+1}-X_{i})_+\big] + \sum_{i=1}^{T} \bE\big[(\mu_{i}-X_{i-1})_+\big] \nonumber \\
&= (\mu_1-\mu_{T+1})+ 3\cdot \sum_{i=2}^{T}\bE\big[(\mu_{i}-X_{i-1})_+\big]+ \bE[(\mu_{1}-X_{0})_+]+ 2 \cdot\bE[(\mu_{T+1}-X_{T })_+]    \nonumber \\
&= 2 \cdot \bE[(\mu_{1}-X_{0})_+] + 3\cdot \sum_{i=2}^{T}\bE\big[(\mu_{i}-X_{i-1})_+\big]
\leq 3\cdot \sum_{i=1}^{T}\bE\big[(\mu_{i}-X_{i-1})_+\big] . \nonumber
\end{align}
Hence, $\bE[\opt] \leq 3 \cdot \sum_{i=1}^{T}\bE[(\mu_i-X_{i-1})_+]$. 

Finally, in the case that $\mu_{i}=\mu$ for all $i\in[1,T]$, observe that we can remove one of the $\bE[(\mu_{i+1}-X_{i})_+]$ terms from the inequality and obtain $\bE[\opt] \leq 2 \cdot \sum_{i=1}^{T}\bE[(\mu-X_{i-1})_+]$. This concludes the proof of \Cref{lem:opt-basic}.
\end{proof}
\subsection{A Tight Lower Bound}
In this section we prove the second part of \Cref{thm:main1},
\begin{proposition}\label{prop:lower}
    For any $\eps< 1$, there is a large enough $T$, and a sequence of distributions $X_1, \ldots, X_T$ such that for any online algorithm for the trading prophets problem it holds that
    \[\bE[\opt] \geq \left(3-O(\eps)\right)\cdot \bE[\alg]\]
    Similarly, there is a sequence of i.i.d.\ distributions such that
    \[\bE[\opt] \geq \left(2-O(\eps)\right)\cdot \bE[\alg]\]
\end{proposition}

We first prove the following lemma. 
\begin{lemma}\label{lem:loweralg}
The expected profit of any online algorithm for the trading prophets problem is at most $\sum_{i=1}^{T}\bE[(\mu_{i}-X_{i-1})_+]$.  
\end{lemma}

Together with \Cref{lem:alg-basic}, this implies that the expected profit of \Cref{alg:buylsellh} is at least as large as the expected profit of any other online algorithm.
We note that the proof of \Cref{lem:loweralg} generalizes that of \cite[Lemma 2.2]{RCV25} to the adversarial model. 

\begin{proof}[Proof of \Cref{lem:loweralg}]
Consider an algorithm that buys the stock at times $b_1, \ldots, b_k$ and sells the stock at times $s_1, \ldots, s_k$ such that $b_1 \leq s_1 \leq b_2 \leq s_2 \leq 
\ldots \leq b_k \leq s_k$. We can equivalently rewrite the algorithm as one that, for every $i \in [k]$, buys at every time step in the interval $[b_i, s_i-1]$ and sells at every time step in the interval $[b_i+1, s_i]$. Hence, we may assume that every algorithm always sells one time step after buying, no matter the selling price. The price at time $i+1$ is independent of the price at time $i$ and of the decision of the algorithm at time $i$. Therefore, if the realizations of the prices are $x_1, \ldots, x_T$, and $p_i$ is the probability that the algorithm buys the stock at time $i$ given the realization $x_i$, then the algorithm's profit is
\[\bE[\alg] = \sum_{i=1}^T p_{i-1} (\bE[X_i] - x_{i-1}).\]
This expression is maximized by the  algorithm that buys at step $i$ if and only if $x_{i} \leq \bE[X_{i+1}]$, and in this case the expected profit is $\sum_{i=1}^{T}\bE[(\mu_{i}-X_{i-1})_+]$.
\end{proof}
We are now ready to prove the proposition, which is a slight adaptation of the examples in \cite[Propositions $1$ and $2$]{CDHOS23}.

\begin{proof}[Proof of \Cref{prop:lower}]
For the first statement, fix $\eps <1$ and consider the following sequence of distributions
\begin{align*}
   X_0& =0 \\
    X_{2i-1} & = \begin{cases}\nicefrac{1}{\eps} \phantom{-11} & \text{w.p. } (1-\eps),\\
    0 & \text{w.p. }\eps.\end{cases} & \forall i\in [T/2]\\
    X_{2i} & = \begin{cases}\nicefrac{1}{\eps}-1 & \text{w.p. }(1-\eps),\\
    \nicefrac{2}{\eps}-1  & \text{w.p. }\eps.\end{cases} &
\end{align*}
These random variables are non-negative and have expectations $\mu_{2i-1}=\nicefrac{1}{\eps}-1$ and $\mu_{2i}=\nicefrac{1}{\eps}$. Simple calculations yield
\begin{eqnarray*}
\bE[(\mu_{2i+1}-X_{2i})_+] & = & \bE[(\nicefrac{1}{\eps}-1-X_{2i})_+] =0,\\
    \bE[(\mu_{2i}-X_{2i-1})_+] & = & \bE[(\nicefrac{1}{\eps}-X_{2i-1})_+] = \eps \cdot \nicefrac{1}{\eps} = 1,\\
    \bE[(X_{2i}-X_{2i-1})_+] & =  & \eps \cdot (1-\eps) (1/\eps -1) + \eps \cdot \eps (2/\eps -1) + (1-\eps) \cdot \eps (1/\eps-1) \\ 
    & = & 2(1-\eps)^2 +\eps (2-\eps) \geq 2-4\eps,\\
    \bE[(X_{2i-1}-X_{2i})_+] & = & (1-\eps) \cdot (1-\eps) \cdot 1 \geq 1-2\eps,\\
    \bE[(X_{1}-X_{0})_+] & = & \bE[(\mu_{1}-X_{0})_+] = \nicefrac{1}{\eps}-1. 
\end{eqnarray*}
By \Cref{obs-opt} and \Cref{lem:loweralg},  for small enough $\eps>0$, 
\begin{align*}
    \bE[\opt]& = \sum_{i=1}^{T}\bE[(X_{i}-X_{i-1})_+] \geq  \frac{T}{2}\left(3- 6\eps\right) + (\nicefrac{1}{\eps}-1) - (1-2\eps) \geq \frac{T}{2}\left(3- 6\eps\right), \\
    \bE[\alg] &\leq \sum_{i=1}^{T}\bE[(\mu_{i}-X_{i-1})_+]  = T/2 + \nicefrac{1}{\eps}-1 \leq T/2 + \nicefrac{1}{\eps}.
\end{align*}
Taking $T$ sufficiently large concludes the proof of the first statement.

For the second statement, consider the following distribution (as in \cite{CDHOS23}): $X_0=0$ and  
\begin{align*}
    X_{i} & = \begin{cases}1& \text{w.p. }\frac{\eps}{2},\\
    \frac{1}{2}& \text{w.p. }1-\eps,\\
    0 & \text{w.p. }\frac{\eps}{2}.\end{cases} & \forall i\in [T]
\end{align*}
This time $\mu=\frac{1}{2}$ and 
\begin{align*}
      \bE[\opt]& = \sum_{i=1}^{T}\bE[(X_{i}-X_{i-1})_+] = \frac{1}{2} + T \cdot \left(\frac{\eps}{2}(1-\eps)\cdot\frac{1}{2} + \left(\frac{\eps}{2}\right)^2\cdot 1 + (1-\eps)\cdot\frac{\eps}{2}\cdot \frac{1}{2}\right) \geq T\cdot \left(\frac{\eps}{2} -\frac{1}{4}\eps^2\right), \\
    \bE[\alg] &\leq \sum_{i=1}^{T}\bE[(\mu-X_{i-1})_+]  = \frac{1}{2} + T \cdot \frac{\eps}{4}.  
\end{align*}
Once again taking $T$ sufficiently large concludes the proof.
\end{proof}

\section{Trading Prophets with Transaction Fees in the IID Model}\label{sec:transactions}
In this section we analyze \Cref{alg:trnsctcost}, proving \Cref{thm:main2}. First, define
\begin{itemize}
    \item $p\triangleq \Pr[X\leq z_L]=Pr[X\geq z_H]$,
    \item $v\triangleq z_H(1-\epspi)-\epssig = z_L(1+\epspi)+\epssig$,
    \item $v_{L}\triangleq \bE[X \mid X\leq z_L]$,
    \item $v_{H}\triangleq\bE[X \mid X\geq z_H]$.
\end{itemize}
Note that $z_L \le v \le z_H$. We prove the theorem by upper bounding \opt and lower bounding \alg.

\begin{lemma}\label{lem:alg-trnsct}
    \Cref{alg:trnsctcost} satisfies $\bE[\alg]  \geq \frac{1}{2} \cdot pT \cdot \left[v_H(1-\epspi)- v_L(1+\epspi)-2\epssig\right]$.
\end{lemma}

\begin{lemma}\label{lem:opt-trnsct}
It holds that 
    $\bE[\opt] \leq v+ pT \cdot \left[v_H(1-\epspi)- v_L(1+\epspi) -2\epssig\right]$.
\end{lemma}

Note, this implies that \Cref{alg:trnsctcost} is $2$-competitive with additive cost $v$. We start by bounding \opt.
\begin{proof}[Proof of \cref{lem:opt-trnsct}]
To upper bound $\bE[\opt]$, we begin by giving extra power to the optimal solution, which can only increase $\bE[\opt]$. In addition to the option of buying the stock at price $X_i(1+\epspi)+\epssig$ and selling the stock at a price of $X_i(1-\epspi)-\epssig$, we allow the optimal offline solution to (a) immediately sell stock purchased in the same round (as well as the  initial inventory stock) for a price of $v$, and (b) to buy stock (just before the current round) at a price of $v$.

 With this change, the strategy of the optimal solution that sees all realizations $X_i$ is simple. At any time $i=1, \ldots, T$, if $X_i\leq z_L$, buy the stock and sell it immediately at a price of $v$. At any time when $X_i \geq z_H$, buy the stock for a price of $v$ and sell it immediately. 

To see that this is the optimal strategy, first observe that we can assume any algorithm with this extra power only ever buys stock that it sells immediately in the same round for price $v$, and only ever sells stock that it bought just prior to the round at price $v$. To this end, consider any solution that buys in round $i$ and sells in round $j > i$. Then we can modify this solution to sell the stock in round $i$ for a price of $v$, and buy the stock just before round $j$ at a price of $v$, without changing the overall profit. 

Next, observe that for round $i$, the profit obtainable from buying and then immediately selling at price $v$ is $v- (X_i(1+\epspi) +\epssig)$, which is nonnegative if and only if $X_i \leq z_L$. Hence, the solution should perform this transaction if and only if $X_i \leq z_L$, otherwise the overall profit can be improved by following this rule. Similarly, the profit obtainable from buying at price $v$ and then immediately selling is $X_i(1-\epspi)- \epssig - v$, which is nonnegative if and only if $X_i \geq z_H$. Once again, the solution should perform this transaction if and only if $X_i \geq z_H$, otherwise profit can be improved by following the rule.

To conclude, assuming the optimal solution uses the strategy outlined above, the total expected profit is
\begin{align*}
\bE[\opt] & \leq  v+  \sum_{i=1}^{T}  \Pr[X \leq z_L] \cdot \bE[v - (X_i(1+\epspi) +\epssig) \mid X \leq z_L]   \\
      & \quad + \Pr[X \geq z_H] \cdot \bE[X_i(1-\epspi)-\epssig - v \mid X \geq z_H]  \\
&=  v+  \sum_{i=1}^{T} p \cdot [v - (v_L(1+\epspi) +\epssig) ]+ p \cdot [v_H(1-\epspi)-\epssig - v ]\\
&= v+ pT \cdot \left[v_H(1-\epspi)- v_L(1+\epspi) -2\epssig\right],
\end{align*}
where the extra $v$ comes from selling the initial stock. 
\end{proof}

We end with the bound on \alg's performance.
\begin{proof}[Proof of \cref{lem:alg-trnsct}]
Let $B_i$ and $S_i$ be indicators that the algorithm buys and sells the stock respectively at time step $i$.
As $X_i$ is independent of the event that the algorithm holds the stock at time step $i$, we have $\bE[X_i \mid S_i=1]= \bE[X_i \mid X_i\geq z_H]=v_H$ and $\bE[X_i \mid B_i=1]= \bE[X_i \mid X_i\leq z_L]=v_L$. 
Moreover, for any time step $i\in[T]$, we have that $\bE[B_i+S_i]=p$. This is true because conditioned on holding the stock, the algorithm sells with probability $p$, and conditioned on not holding the stock, the algorithm buys with probability $p$. Defining $B=\sum_{i=1}^{T}B_i$ to be the total number of purchases and $S=\sum_{i=1}^{T}S_i$ to be the total number of sales, we get $\bE[B+S]=pT$. Since the algorithm initially owns the stock, we have $S-1 \le B\leq S$, which implies that $\bE[S]\geq \frac{pT}{2}$ and $\bE[B] \leq \frac{pT}{2}$.  Hence, 
\begin{align*}
    \bE[\alg] & = \sum_{i=1}^{T} \Pr[S_i=1]\cdot \bE[X_i(1-\epspi)-\epssig \mid  S_i=1] - \Pr[B_i=1]\cdot \bE[X_i(1+\epspi)+\epssig \mid  B_i=1] \\
    & = \sum_{i=1}^{T}\bE[S_i]\cdot \left[v_H(1-\epspi)-\epssig\right] - \bE[B_i]\cdot \left[v_L(1+\epspi)+\epssig\right] \\ &= 
    \bE[S]\cdot \left[v_H(1-\epspi)-\epssig\right] - \bE[B]\cdot \left[v_L(1+\epspi)+\epssig\right] \\
    & \geq \frac{pT}{2} \cdot \left[v_H(1-\epspi)- v_L(1+\epspi)-2\epssig\right]. \qedhere
\end{align*}
 
\end{proof}

\section{Conclusion}

In this work, we show that allowing an additive term in the competitive ratio, or simply permitting the algorithm to start with an initial stock, bypasses the impossibility theorem of prior work and enables competitive algorithms for the trading prophets problem. In particular, we present an optimal $3$-competitive algorithm for this original problem as well as a $2$-competitive algorithm for a harder version with transaction costs but i.i.d.\ prices.

This raises a natural open question: can we design a competitive algorithm for the trading prophets problem with transaction costs in the adversarial (i.e. non-i.i.d.) setting? Another interesting direction is to generalize our results to the multiple commodity setting studied by \cite{RCV25}, or to investigate intermediate versions that bridge between ours and the fully adversarial (and significantly harder) model of \cite{azar2025competitivebundletrading}.

\bibliographystyle{alpha}
\bibliography{bib}

\appendix
\section{Barriers for the Non-i.d.d.\ Model with Transaction Costs}\label{sec:bad-example}

In this section we explore the trading prophets problem with fixed additive transaction costs $\eps>0$.
We show that several natural approaches fail.

\begin{theorem}
Consider the trading prophets problem with fixed additive transaction costs of $\eps>0$.
Then, any online algorithm  whose decisions at time step $t$ are only based on the prices in the next $k$ time steps for any fixed $k$ cannot have a bounded competitive ratio. This holds even if at time step $t$, the algorithm is given the realizations of the prices in time steps $t+1, \ldots, t+k$.
\end{theorem}

\begin{proof}
We give a deterministic sequence of prices such that the profit of an optimal solution that knows the entire sequence is unbounded, while an online algorithm that bases its decisions at time step $t$ only on the (deterministic) prices at time steps $t+1, \ldots, t+k$ cannot make any non-negative profit. We prove the theorem for $k=1$. For larger, $k$ just duplicate each price $k$ consecutive times. We assume that the transaction cost is additive $\eps$, but the proof holds also for multiplicative transaction cost.

The price sequence is composed of any number of phases. A phase may start at an arbitrary time step $t$, but for clarity, assume that the first time step in the phase is at time step $t=1$.
The price at $t=1$ is $x_1=1$ and the algorithm is given the price at time step $t=2$ which is $x_2=1+\eps$. We maintain the invariant that if the online algorithm holds the stock at the end of time step $t=1$, then the price paid by the algorithm for the stock is at least $1$. If the algorithm buys the stock at $t=1$ or holds the stock after at the end of time step $t=1$, then we reveal two prices $x_3=1-\eps$ and $x_4=1+\frac{3}{2}\eps$. In this case, the optimal solution can buy the stock at $t=3$ and sell it at $t=4$, making profit of $x_4-x_3- 2\eps = \frac{\eps}{2}$. The online algorithm that purchased the stock for a price of $1$ cannot make any profit. It can, in principle, sell the stock making a negative profit, but as the phase ends and the next price is $1$, there is no reason to do so.

If, on the other hand, at the end of time step $t=1$ the algorithm does not hold the stock, then the price at $t=3$ is $x_3=1+\frac{5}{2}\eps$ and the phase ends. The optimal solution can buy the stock at time step $t=1$ and sell it at time step $t=3$, making a profit of $x_3-x_1-2\eps = \frac{\eps}{2}$. If the algorithm buys the stock at time step $t=2$ and sells it at time $t=3$ it makes a negative profit. In any case, as the phase ends and the next price is again $1$ the algorithm has no reason to buy the stock at time $t=2$ at a price $1+\eps$. This sequence of phases can be repeated  any number of times.
\end{proof}

Finally, we show that even for the i.i.d model with fixed transaction costs, a natural generalization of \Cref{alg:buylsellh} fails.

\begin{observation}
Consider the trading prophets problem with fixed additive transaction costs of $\eps>0$, and consider the algorithm that buys whenever the stock is not in inventory and $x_i \leq \mu_{i+1}-\eps$ and sells whenever the stock is in inventory and $x_i\geq \mu_{i+1}+\eps$. Then this algorithm is not competitive, even in the i.i.d setting.
\end{observation}

\begin{proof}
Consider the following distribution of prices.
\begin{align*}
    X_{i} & = \begin{cases}1+2\eps & \text{w.p. } \frac{1}{5},\\
    1-\frac{\eps}{2}&  \text{w.p. }\frac{4}{5}.\end{cases} & \forall i\in [T]
\end{align*}

For this distribution, $\mu= \frac{1}{5}(1+2\eps) + \frac{4}{5}(1-\frac{\eps}{2})=1$. Hence the algorithm never buys and can make a profit of at most $1+2\eps-\eps$ by selling the initial inventory. On the other hand,
\begin{align*}
      \bE[\opt]& \geq \sum_{i=1}^{T-1}Pr[X_i=1-\eps/2, X_{i+1}=1+2\eps]\cdot \big((1+2\eps -\eps) -(1-\frac{\eps}{2}+\eps)\big) \\
      & \geq (T-1) \cdot \frac{4}{25}\cdot \big((1+2\eps -\eps) -(1-\frac{\eps}{2}+\eps)\big) = (T-1) \cdot \frac{2\eps}{25}.
\end{align*}
Thus, for a large enough $T$, the algorithm is not competitive.
\end{proof}

\end{document}